\documentclass[aps,pra,preprint,groupedaddress,showkeys, showpacs]{revtex4-1}
\usepackage{comment,
latexsym,epic,eepic,
graphics,times,epsfig,amssymb, 
afig2
}
\usepackage{amsthm}
\newtheorem{theorem}{Theorem}
\newtheorem{corollary}{Corollary}

\newdimen\tempor

\def\hbt#1{\hbox{\rm\tiny #1}}
\def\hbf#1{\hbox{\rm\footnotesize #1}}
\def\hb#1{\hbox{\rm #1}}
\def\sottovar(#1)#2{%
      \mathop{\vtop{\ialign{##\crcr%
$\hfil\displaystyle{#2}\hfil$\crcr\noalign{\kern#1\nointerlineskip}
      \hfill\crcr\noalign{\kern#1}}}}\limits}

\def\qmatrix#1{\left[ \matrix{#1} \right]}

\def\C(#1){{\cal #1}}
\def\M(#1){\mathbb #1}

\def\vq{\;,\qquad}
\def\PD(#1,#2){{\rm PD$(#1,#2)$}}
\def\met{{1 \over 2}}
\def\bk(#1,#2){\langle#1|#2\rangle}
\def\bok(#1,#2,#3){\langle#1|#2|#3\rangle}
\def\kb(#1,#2){|#1\rangle\langle#2|}
\def\out#1{|#1\rangle\langle#1|}

\def\Ld{\Lambda}
\def\ket#1{|#1\rangle}

\def\Tr{{\rm Tr}}
\def\diag{\mathop{\rm diag\;}\nolimits}

\def\I{i}
\def\E{\hb{e}}
\def\TT{{\hbt{T}}}
\def\Bas{a_{\displaystyle*}}

\def\inter#1#2{#1\kern-.0667em \lower.5ex\hbox{$#2$}}

\begin{document}

\title{Gaussian States and the Geometrically Uniform Symmetry}
  \author{Gianfranco Cariolaro}
  \author{Roberto Corvaja}
  \author{Gianfranco Pierobon}
\affiliation{Department of Information Engineering,
University of Padova, \\
Via G. Gradenigo 6/B - 35131 Padova, Italy
}


\date{\today}

\begin{abstract}
Quantum Gaussian states can be considered as the majority of the practical quantum states used in quantum communications and more generally in quantum information. Here we consider their properties in relation with the geometrically uniform symmetry, a property of quantum states that greatly simplifies the derivation of the optimal decision by means of the square root measurements. In a general framework of the $N$-mode Gaussian states we show the general properties of this symmetry and the application of the optimal quantum measurements. An application example is presented, to quantum communication systems employing pulse position modulation. We prove that the geometrically uniform symmetry can be applied to the general class of multimode Gaussian states.
\end{abstract} 


\maketitle

\section{Introduction}

In the last years Gaussian states have received a tremendous interest \cite{Weedbrook}--\cite{milanesi}, due to the fact  that most quantum operations can be performed with continuous variables, of which Gaussian states represent the most important class. The advantage with respect to discrete variables (qubit) is that the optical implementation of continuous variables is available and robust. A relevant application of Gaussian states is given by Quantum Communications (QC), which in practice are implemented by coherent states, the most important subclass of Gaussian states. Often in QC coherent states are not related explicitly to Gaussian states \cite{Kato1999}--\cite{CarPier}, but their recent developments lead to a very elegant and powerful theory, so that it seems important to revisit QC using this theory.

In this context we reconsider QC having as information carrier Gaussian states and we assume that a finite set (constellation)
of Gaussian states has a special form of symmetry, called {\it geometrically uniform symmetry} (GUS). A constellation having the GUS is generated starting from a single reference state through a unitary operator, called {\it symmetry operator}. The GUS has the advantage, not only of simplifying the theory of QC, but also to derive the optimal decision measurements, otherwise not possible. It is worth to remark that the standard form of GUS is enjoyed by almost all constellations considered for practical QC, namely {\it phase shift keying} (PSK) and {\it pulse position modulation} (PPM). Also {\it quadrature amplitude modulation} (QAM) verifies a generalized form of GUS.

Since we want to establish completely general results, valid for multimode Gaussian states, a fundamental preliminary is a clear and compact formulation of the most general Gaussian state and of the most general Gaussian transformation (Gaussian unitary). To this end we follow the theory developed by Ma and Rhodes in a seminal paper published in 1990 \cite{Ma}. This theory has the advantage to handle, through an appropriate algebra of operators, the general $N$-mode Gaussian states in much the same way as the
single--mode Gaussian states. Substantially it proves that the most general Gaussian unitary is given by a cascade combination of a squeezing,
a displacement, and  a rotation. Finally we will prove that, starting from an arbitrary $N$-mode Gaussian state, we can generate a GUS constellation, where the symmetry operator is provided by the rotation operator with an appropriate amount of rotation.

In the literature only QC systems using coherent states have been considered and little attention has been devoted to the appealing possibility of using squeezed states \cite{Slusher}. To this end we revisit QC, where the main problem is the optimization of quantum detection to achieve the minimum error probability. As known, this problem is very difficult and
exact solutions are established only in few cases. To overcome this difficulty, suboptimal solutions are considered,  the most important of which is given by the square-root measurements (SRM), introduced by Hausladen {\it at al.} \cite{Hausladen} and  subsequently developed as least square measurement (LSM) by Eldar and Forney \cite{Forney1} \cite{Eldar3}.
This technique is not in general optimal, but gives a good approximation of the optimum (``pretty good'' is the judgment given by the authors and very often echoed in the literature). The SRM/LSM can be applied with any constellation, but in the presence of GUS it provides the optimal solution in an easy and explicit form. This holds when the detection is based on pure states; with mixed states the SRM technique is suboptimal, but gives a very accurate overestimate of the error probability, surely better than the quantum Chernov bound usually considered with Gaussian states \cite{Corvaja}. 
 
The paper is organized as follows. In Section~\ref{Gauss} we introduce the $N$-mode Gaussian states and Gaussian transformations and discuss their most general representations in terms of unitary Gaussian transformations. Section~\ref{GUS} is devoted to  Gaussian states equipped with the GUS. In Section~\ref{Applications} we present an application of the theory to the quantum detection for PPM, which requires a non trivial analysis in a multimode Hilbert space. 
Explicit examples of error probability are carried out considering as quantum carrier in PPM both coherent and squeezed states. 

We adopt the following notation: $\cdot^\TT$ denotes transposition, while $\cdot^*$ has the multiple role of adjoint for operators, Hermitian conjugation for matrices, and complex conjugation for numbers. 
The {\bf normalization} follows the notation of \cite{Weedbrook}, where in particular the reduced Planck constant is set to $\hbar=2$. 

\section{Gaussian states and Gaussian transformations}\label{Gauss}
\subsection{Definition of Gaussian states}

In an $N$-mode bosonic space $\C(H)^{\otimes N}$ quantum states are represented
in general by density operators 
$\rho:\,\C(H)^{\otimes N}\;\to\;\C(H)^{\otimes N}$. Any density operator
has an equivalent representation in the {\it phase space} $\M(R)^{2N}$
given by a characteristic function and Wigner function. 
In particular, Gaussian states are defined with reference to their characteristic and Wigner functions, which should have a multivariate Gaussian form. For an $N$-mode Gaussian state with mean $\overline{X}$ and covariance matrix $V$ the characteristic function has the following form
\begin{equation}
	\chi(u,v) = \exp\left[-\frac{1}{2}\inter{u}{v}^\TT \Omega V \Omega^\TT \inter{u}{v} -j\,\left(\Omega \overline{X}\right)^\TT \inter{u}{v}\right]\;, \qquad (u,v)\in\mathbb{R}^{2N}
\label{T1}
\end{equation}
with 
\begin{equation}
	\inter{u}{v}= [u_1, v_1, \ldots,u_N, v_N]^T\;,\qquad \Omega=\diag [\Omega_1, \ldots ,\Omega_N]\;,\qquad \Omega_i =\left[\begin{array}{cc} 0 &1 \\ -1 &0 \end{array}\right]
	\label{T2}
\end{equation}
The relevant property of Gaussian states is that they are specified simply by the pair $(\overline{X},V)$ and for this reason the density operator of a Gaussian state is often indicated in the form $\rho(\overline{X},V)$.

\subsection{Gaussian transformations}

A {\it quantum transformation} or {\it quantum operation} maps the state of the system $\rho$ into a new state $\tilde\rho=\Phi(\rho)$.
In general  a quantum transformation defines a quantum channel, which may refer to an open system \cite{HolevoGiovannetti}, while 
in closed quantum systems the map is provided by a unitary transformation according to
\begin{equation}
	\rho\,\to\,\tilde \rho=U\,\rho\,U^*\;. \label{T40}
\end{equation}
A quantum transformation is {\bf Gaussian} when {\it it transforms Gaussian states into Gaussian states}.
When the Gaussian  transformation is performed according to the unitary map (\ref{T40})
 it is called {\bf Gaussian unitary}. It can be shown \cite{Weedbrook} that Gaussian unitaries are generated
in the form $U=\exp(-\I H/2)$, where $H$ is a Hamiltonian, which
is a second--order polynomial in the field operators $\inter qp:=[q_1,p_1,\ldots,q_N,p_N]^\TT$ or, equivalently, in the bosonic
operators $a:=[a_1,\ldots,a_N]^\TT$, $\Bas:=[a^*_1,\ldots,a^*_N]^\TT$. 
 
In terms of quadrature operators $\inter qp$, a Gaussian unitary gives a {\bf symplectic transformation}, which has the form
\begin{equation} 
	\inter qp\to S\,\inter qp+ d \label{T42}
\end{equation}
where $S$ is a $2N\times2N$ real matrix  and $d\in\M(R)^{2N}$. $S$ has the property (symplectic matrix) $S\,\Omega\,S^{\hbt{T}}=\Omega$.

A symplectic  transformation  modifies the mean vector $\overline X=\overline{\inter qp}$ and the covariance matrix $V$ in the form
\begin{equation}
    \overline{X}\;\to \;  S\;  \overline{X}+d\vq V\;\to\;S\,V\,S^{\TT}\;. \label{T60}
\end{equation}
These are the key results because they allow us to specify a Gaussian transformations in terms of the parameters $(S,d)$, which ``live'' in the phase space $\M(R)^{2N}$. 
 
\subsection{Fundamental Gaussian unitaries}

In the literature we find a plethora of forms for the Gaussian unitaries with specific expressions in the single-mode, in the two-mode and in the $N$-mode. Here we follow the unified form developed by Ma and Rhodes \cite{Ma} for the $N$-mode.
This form, using appropriate matrix notations, turns out to have
extremely similar algebraic properties as that of the single mode and is very useful to establish general results.

There are only three fundamental Gaussian unitaries, which are  specified by the following unitary operators: 
\begin{enumerate}
\item $N$-mode displacement operator
\begin{equation}
    D(\alpha):=\E^{a^* \alpha\,\,-\alpha^*\, a} 
    \vq \qquad\alpha=[\alpha_1,\ldots,a_N]^\TT \in \M(C)^N\;. \label{U22}
\end{equation}
\item $N$-mode rotation operator
\begin{equation}
   R(\phi):=\E^{\I\, a^*\phi \, a}\vq \qquad\phi \quad\hb{$N\times N$ Hermitian matrix}\;. \label{U24}
\end{equation}
\item $N$-mode squeeze operator
\begin{equation}
   S(z):=\E^{\met\left[\,( a^* \,z\,{ \Bas}-a^\TT\,z^*\,a)\right]} 
   \vq z\quad \hb{$N\times N$  symmetric  matrix}\;. \label{U26}
\end{equation}
\end{enumerate} 
In these definitions $a=[a_1,\ldots,a_N]^\TT$ and $\Bas=[a_1^*,\ldots,a_N^*]^\TT$ are column vectors, while $a^*=[a_1^*,\ldots,a_N^*]$ is a row vector. In (\ref{U24}) the $N\times N$ symmetric matrix $z$ can always be written in the forms \cite{Ma} $z=r\,\E^{\I \theta}=\E^{\I \theta^\TT}\,r^\TT$, 
where $r$ and $\theta$  are $N\times N$ Hermitian (in general non commuting) matrices and $r$ is positive semidefinite. 
A particular case of (\ref{U24}) with $N=2$ gives the beam splitter, while the most popular forms of squeezing (single mode and two-mode) are particular cases of (\ref{U26}) for $N=1, 2$.

Note that the above fundamental unitaries are special cases of the general Gaussian unitary $U=\E^{-\I H/2}$ with $H$ a Hamiltonian, quadratic in the creation and annihilation operators collected in $a_*$ and $a$.
As we shall see below, all Gaussian transformations are obtained as combination of these operators, and the corresponding
Gaussian states are typically generated starting from replicas of vacuum states or of coherent states.

We are particularly interested in the cascade combination, where one can switch the order of operators with an appropriate change in the parameters \cite{Ma}
\begin{eqnarray}
  &&D(\alpha)\,S(z)=S(z)\,D(\beta)\!\!\!\!\quad\quad           \beta=\cosh(r)\,\alpha- \sinh(r)\, \E^{\I\theta}\,\alpha_* \\ \label{SRa}
  &&S(z)\,R(\phi)=R(\phi)\,S(z_0)\!\!\!\!\quad\;\; \;\,\,   z_0=\E^{-\I \phi}z\, \E^{-\I \phi^{\tt T}} \\ \label{SRb}
  &&D(\alpha)\,R(\phi)=R(\phi)\,D(\beta)\!\!\!\!\quad\;\,\, \beta= \alpha\,\E^{-\I \phi}\;.\label{SRc}
\end{eqnarray}

\subsection{The most general Gaussian unitary}

The importance of the fundamental unitaries lies in the following:
\begin{theorem}\label{P5} 
The most general Gaussian unitary is given by the cascade combination of the three fundamental Gaussian unitaries:
$S(z)$, $D(\alpha)$ and $R(\phi)$, cascaded in any arbitrary order by a proper adjustments of the parameters.
\end{theorem}

The proof can be obtained for the general multimode using the Lie algebra: Ma and Rhodes \cite{Ma}, generalizing a previous result obtained for the single mode \cite{Shumaker86}--\cite{Ma89}, proved that a unitary operator $\E^{-\I H/2}$, where $H$ is a general $N$-mode quadratic Hamiltonian, can be written in the form
\begin{equation}
	U=\E^{\I\,\gamma}\,S(z)\,D(\alpha)\,R(\phi)\label{K16}
\end{equation}
where the phasor $\E^{\I\,\gamma}$ with $\gamma\in{\mathbb R}$ is irrelevant for the state generation. On the other hand we can apply the switching rules (\ref{SRa}--\ref{SRc})
to change the order of the fundamental unitaries in (\ref{K16}), with appropriate modifications of the parameters. 

In the phase space a Gaussian unitary is equivalent to a symplectic map (\ref{T42}), specified by the pair $(S,d)$. The Gaussian unitary  can always be written in the form   $U_{S,d}= D(\alpha) U_S$, where $U_S$ corresponds to the map $\inter qp\to S\inter qp$ and the displacement operator $D(\alpha)$ in the phase space provides the displacement $\inter qp\to \inter qp+d$, with $d_{2i-1}=\Re\,\alpha_i$ and $d_{2i}=\Im\,\alpha_i$, 
 
\subsection{The most general Gaussian state}
    
The most general Gaussian state can be derived by combination of the {\it thermal decomposition} and Theorem~\ref{P5}.
 
In a Gaussian state with the pair $(V,\overline X)$, the covariance matrix $V$ and the mean vector $\overline X$ can be handled separately. The covariance matrix $V$ is fully described by powerful Williamson's theorem, which states that an $N$-mode covariance matrix $V$ can be decomposed in the form
\begin{equation}
\fbox{$\displaystyle V= S_w\,V^\oplus\,S_w^\TT 
 \vq\quad V^\oplus=\diag[\sigma^2_1,\sigma^2_1,\ldots,\sigma^2_N,\sigma^2_N]$} \label{W2}
\end{equation}
where $S_w$ is a $2N\times 2N$ symplectic matrix and the $\sigma^2_i$ are positive real values, called the {\bf symplectic eigenvalues} of $V$.

Application of  Williamson's theorem gives the so called {\it thermal decomposition} of a Gaussian state \cite{milanesi}, that is, an arbitrary $N$-mode zero--mean Gaussian state can be generated by the {\bf tensor product of $N$ single-mode thermal states}, with covariance matrix $V_k=\sigma^2_k\,I_2$, where $I_2$ is the $2\times 2$ identity, and number of thermal photons $\C(N)_k=\met(\sigma^2_k-1)$.
In fact, a single--mode thermal state is specified by a covariance matrix $\sigma^2I_2$, with average photon number $\C(N)=\met(\sigma^2-1)$. Now, according to  (\ref{W2}), the $N$-mode Gaussian state with the diagonal covariance matrix $ V^\oplus$, is given by
\begin{equation}
   \rho_{\hbf{th}}(0,V^\oplus)=
      \rho(0,\sigma_1^2\,I_2)\otimes\cdots\otimes
      \rho(0,\sigma_N^2\,I_2)\;.\label{W8}
\end{equation}
However, by Theorem~\ref{P5}, we know that $U_S$ can be written as the cascade combination $U_S=S(z)\,R(\phi)$ or $U_S=R(\phi)\,S(z)$.
Then, a zero--mean Gaussian state with covariance matrix $V$ is generated from $\rho_{\hbf{th}}(0,V^\oplus)$ in the form
$\rho(0,V)= U_S\,\rho_{\hbf{th}}(0,V^\oplus)\,U_S^*$, 
where $U_S$ is the unitary operator corresponding to the symplectic transformation $S_w$ of decomposition (\ref{W2}). 

A Gaussian state with {\bf non--zero mean} is generated by introducing in $\rho(d,V)$ an appropriate displacement operator. 
In conclusion, by combination of the previous statements:
\begin{theorem}\label{H4}
The most general $N$-mode Gaussian state is generated
from thermal state (\ref{W8}) by application of the three fundamental unitaries
as
\begin{equation}
	\rho(d,V)= D(\alpha)\,R(\phi)\,S(z)\,\rho_{\hbf{th}}(0,V^\oplus) \,S^*(z)\,R^*(\phi)\,D^*(\alpha) \label{W12}
\end{equation}
where the order $D(\alpha)\,R(\phi)\,S(z)$ can be permuted according to (\ref{SRa})--(\ref{SRc}).
\end{theorem}
   
For the particular case of pure states, the thermal decomposition degenerates into the product of $N$ replicas of the vacuum state, say $\ket{0_N}$, and then $\rho(d,V)=\out{\psi(d,V)}$, with $\ket{\psi(d,V)}=D(\alpha)\,R(\phi)\,S(z)\,\ket{0_N}$.
But we can invert the order of squeezing and rotation with the rules (\ref{SRa})--(\ref{SRc}) and after the change, $R(\phi)\ket{0_N}=\ket{0_N}$. In conclusion:
\begin{corollary}\label{H5} The most general $N$-mode pure Gaussian state is obtained from the $N$ replica of the vacuum $\ket{0_N}$ as
\begin{equation}
	\ket{\psi(d,V)}=D(\alpha)\,S(z)\,\ket{0_N}\, :=\ket{z, \alpha}\;. \label{W14} 
\end{equation}
 \end{corollary}
In words, the most general $N$-mode Gaussian pure state is a squeezed--displaced state or a displaced--squeezed state.

\section{The geometrically uniform symmetry  (GUS) with Gaussian states}\label{GUS}

The context of GUS is provided by {\it QC systems} where the transmission of classic information uses quantum states as physical carriers. A classical source emits a symbol $A$
belonging to a set of $K$ elements, $A\in\C(A)=\{0,1,\ldots,K-1\}$, with assigned {\it a priori} probabilities
$q_i=P[A=i]$, $ i\in\C(A)$.
  The transmitter (Alice) encodes the symbol $A$ into a quantum state $|\gamma_A\rangle$ in a {\it constellation} of $K$ pure states
$\left\{|\gamma_{0}\rangle,\ldots, |\gamma_{K-1}\rangle\right\}$, and more generally, in a constellation of mixed states with 
density operators $\{\rho_{0},\ldots,\rho_{K-1}\}$.

The receiver (Bob) performs a quantum measurement from the received state $\rho_A$ with POVM measurement operators $\{P_k, k\in\C(A)\}$. 
On the basis of the measurement Bob estimates the state sent by Alice. The {\it correct decision probability} is \cite{Helstrom2}
\begin{equation}
	P_c=\sum_{i\in\C(A)}q_i\Tr(\rho_iP_i)\;. \label{B42}
\end{equation}
The choice of the measurement operators maximizing $P_c$ is a (generally difficult) key problem in QC.
In particular, if Alice uses pure states, i.e., $\rho_i=\out{\gamma_i}$, according to Kennedy's theorem \cite{Kennedy}, the optimal POVM have rank one, $P_i=\out{\mu_i}$, where the $\mu_i$ are called {\it measurement 
 vectors} and (\ref{B42}) becomes
\begin{equation}
	P_c=\sum_{i\in\C(A)}q_i\left|\bk({\mu_i},{\gamma_i})\right|^2\;. \label{B43}
\end{equation}
	 
\subsection{Definition of GUS}

Since the beginning of QC \cite{Helstrom2}\cite{Yuen} particular attention has been paid to constellations enjoying a high degree of symmetry with uniform {\it a priori} probabilities. The interest of this case resides both in the fact that it corresponds to many practical situations and that the optimal measurements is easy to obtain \cite{Kato1999}--\cite{CarPier}\cite{Forney1}\cite{Ban1997}.
We now define the GUS for pure states and we assume equiprobable symbols, $q_i=1/K$, but the definition can be extended to generic a priori probabilities $q_i$ substituting the states with the weighted states $\sqrt{q_i}|\gamma_i\rangle$ or  $q_i\,\rho_i$ (see \cite{Forney1}).

A constellation of $K$ pure states $
  \left\{|\gamma_0\rangle,|\gamma_1\rangle,\ldots,|\gamma_{K-1}\rangle\right\}$
has the {\it geometrically uniform symmetry} when the two properties are verified:
 1) there exists a unitary operator $Q$ with the property $Q^{K}=I_{\C(H)^{\otimes N}}$, where $I_{\C(H)^{\otimes N}}$ is the identity operator of $\C(H)^{\otimes N}$, and 
 2) the $K$ states $|\gamma_i\rangle$ are obtained from
 a single reference state $|\gamma_0\rangle$ in the following way
\begin{equation}
	|\gamma_i\rangle=Q^{i}|\gamma_0\rangle\vq i=0,1,\ldots,K-1\;. \label{C36}
\end{equation}
The operator $Q$, which is given by a $K$-th root of the identity operator, is called {\it symmetry operator}.
Thus, in the presence of the GUS, the specification of the constellation is limited to  the symmetry operator $Q$ 
and to the reference state $|\gamma_0\rangle$. In addition, it simplifies the quantum decision, because we can choose the POVMs of the form $P_i=\out{\mu_i}$, where the $\mu_i$  (measurement vectors) have the same symmetry as the states, that is, $|\mu_i\rangle=Q^{i}|\mu_0\rangle$, $i=0,1,\ldots,K-1$.

\subsection{The GUS with Gaussian states }
 
We now investigate the possibility that a constellation of Gaussian states have the GUS. 
Let $\C(S)=\{\ket{\psi(p)}, p\in {\cal P}\}$ be a class of pure quantum states, dependent on a parameter $p$. The class is {\bf closed with respect to rotations} if $R(\phi)\ket{\psi}\in\C(S)$, where $ R(\phi)$ is the rotation operator.
With such a class we can construct constellations of any order $K$ with the GUS property. In practice in the single mode we get a $K$--ary PSK constellations, by choosing an arbitrary reference state $\ket{\psi_0}$ in $\C(S)$  and using as symmetry operator $Q=R(2\pi/K)$. In the multimode a relevant application is given by the PPM (see below). 
We know that the new state is still Gaussian, but we want to find the new parameters $\alpha$ and $z$ determined by the rotation.
We apply relation (\ref{SRc}) to get $R(\phi)\,D(\alpha)=D(\E^{\I\,\phi}\,\alpha)\,R(\phi)$.
Next we apply (\ref{SRb}) to get
$R(\phi)\,S(z)=S(\E^{\I\,\phi}\,z\,\E^{\I\,\phi^\TT}) \,R(\phi)$.
Hence 
\begin{equation}
    \ket{\psi(z,\alpha,\phi)}=D(\E^{\I\,\phi}\,\alpha)\;S(\E^{\I\,\phi}z\,\E^{\I\,\phi^\TT}) \;R(\phi)\;\ket{0_N}:=\ket{z,\alpha,\phi}\;.\label{FT1}
\end{equation}
But $R(\phi)\ket{0_N}=\ket{0_N}$, so that the rotation can be dropped. In conclusion the rotation
modifies the parameters in the form
\begin{equation}
     z=\;\to\;\E^{\I \phi }z\,\E^{\I \phi}\vq
        \alpha\;\to\;\alpha\,\E^{\I \phi}\;, \label{FT3}
\end{equation}
and we have:
\begin{theorem}\label{PP2} 
The class of pure displaced--squeezed states is closed under rotations. A {\emph rotated--displaced--squeezed state}
can be obtained from a {\emph displaced--squeezed state} by modification of the squeeze factor and of the displacement amount as
\begin{equation}
\framebox{$\displaystyle
      \ket{z,\alpha,\phi} = \ket{\E^{\I \phi }z\,\E^{\I \phi^\TT},\E^{\I \phi}\,\alpha}\;.
         $} \label{R32}
\end{equation}
\end{theorem}
With reference to the class $\C(S)=\{\ket{\psi(p)},p\in P\}$ the statement of Theorem~\ref{PP2} can be formulated as follows. The class $\C(S)$ becomes explicitly the class of squeezed displaced states with the correspondence
\begin{equation}
	p=(z,\alpha)\vq P=\M(C)^2 \vq \ket{\psi(p)}=\ket{(z,\alpha)}\;.
\end{equation}
If $p_0=(z_0,\alpha_0)\in\M(C)^2$ is an arbitrary value of $p$, after the rotation, the parameter becomes
\begin{equation}
p_\phi=	(z_\phi,\alpha_\phi)= (\E^{\I \phi }z\,\E^{\I \phi}\;,\; \alpha\,\E^{\I \phi})\;.
\end{equation}
The statement can be reformulated also in the phase space as follows. Let $ V(p)=V(z,\alpha)$ be the covariance matrix 
of the squeezed displaced state $\ket{(z,\alpha)}$, then the rotation provides the change
\begin{equation}
 V(z_0,\alpha_0)\quad\to\quad V(z_\phi,\alpha_\phi)=
         S_{\hbf{rot}}(\phi)\; V(z_0,\alpha_0)\;S_{\hbf{rot}}^\TT(\phi) \label{Vrot}
\end{equation}
where $S_{\hbf{rot}}(\phi)$ is the symplectic matrix of the rotation transformation, which is given by \cite{Ma} 
\begin{equation}
	S_{\hbf{rot}}(\phi)= \qmatrix{\cos \phi&-\sin \phi\cr\sin \phi&\cos \phi}
\end{equation}


\subsection{Extension of the GUS to mixed Gaussian states}

First, the definition of GUS can be extended to mixed states as follows. A constellation of $K$ density operators
$\left\{\rho_0,\rho_1,\ldots,\rho_{K-1}\right\}$ has the GUS when the two properties are verified: 1) there exists a unitary
operator $Q$ with the property $Q^{K}=I_{\C(H)}$ and 2) the $K$ density operators $\rho_i$ are obtained from a single reference density operator $\rho_0$ in the following way
\begin{equation} 
   \rho_i=Q^i\rho_0\,(Q^i)^*\vq i=0,1,\ldots,K-1\;. \label{W42}
\end{equation}
This extension is in harmony with the fact that with pure states the density operators become $\rho_i=|\gamma_i\rangle\langle\gamma_i|$. In addition, with the factorization of the density operators, $\rho_i=\gamma_i\gamma_i^*$, relation (\ref{W42})
gives $ \gamma_i=Q^{i} \gamma_0$, which generalizes (\ref{C36}). In the context of optimal decision \cite{Eldar3} the POVMs can be chosen in the form $P_i=\mu_i \mu^*_i$, where the {\it measurement factors} have the symmetry $\mu_i=Q^{i}\mu_0$, $i=0,1,\ldots,K-1$.

In terms of the characterization of Gaussian states, the previous results obtained for pure states cannot be extended straightforwardly to the whole class of mixed Gaussian states. In fact the critical point in the proof of Theorem~\ref{PP2} is represented by the relation $R(\phi)\ket{0_N}=\ket{0_N}$, in which the ground state $\ket{0_N}$ ``absorbs the rotation''. This property does not hold when the ground state is replaced by a general thermal state.

We remind that a general Gaussian channel \cite{HolevoGiovannetti} is completely specified by a triplet $(E, \ell, F)$, 
where $E$ and $F$ are $2N \times 2N$ real matrices and $\ell \in\M(R)^{2N}$. 
A Gaussian channel transforms the mean $\overline{X}$ and the covariance matrix $V$ of an input state $\rho$ in the form
\begin{equation}
    \overline{X}\;\to \;  E^\TT\, \overline{X}+\ell \vq V\;\to\;E^{\TT}\,V\,E + F\;. \label{channelTriplet}
\end{equation}

To get useful results we have to limit the class of mixed state to a suitable subclass of Gaussian states obtained in the following way, which comprises all the cases of interest for the applications.
We suppose that a pure Gaussian state $\ket{\psi(p)}$ is sent through a Gaussian channel specified  by the triplet $(E, \ell, F)$. At the output the noisy state is still Gaussian, but mixed, with a density operator $\rho(\psi(p))$ as in Fig.~\ref{GQ10}.
\begin{figure}[h]
\includegraphics{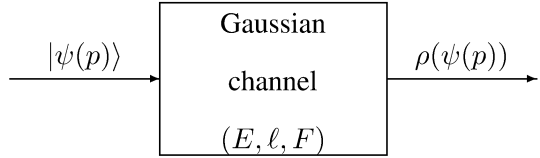}
\vspace*{-2mm}
\caption{A pure Gaussian state $\ket{\psi(p)}$ is sent through a Gaussian channel specified by the triplet $(E, \ell, F)$. The state at the output is still Gaussian, but in general mixed and described by the density operator $\rho(\psi(p))$.}
\label{GQ10} 
\end{figure}
We denote by $\C(S)_{(E, \ell, F)}=\{\rho(\psi), p\in P\}$ this restricted subclass of Gaussian mixed states.

\begin{theorem}\label{PP2ACnew} The class of states $\C(S)_{(E, \ell, F)}=\{\rho(\psi(p))\}$, obtained at the
 output of a Gaussian channel with input pure Gaussian states, is closed under rotations, provided that the matrix
$E$ commutes with the rotation matrix $S_{\hbf{rot}}(\phi)$ and $F$ has the form $f\,I_{2N}$, with $f$ a scalar.
\end{theorem}
\begin{proof} 
The mean vector is modified as $\overline{X}\;\to \;  E^\TT\, \overline{X}+\ell$, while the covariance matrix is modified as 
\begin{equation}
    V\;\to\;E^{\TT}\,V\,E + F\;.
\end{equation}

Consider a generic generic pure state $\ket{\psi(p_0)}$ in the class $\C(S)$ with covariance $V(p_0)$.
Let $V(p_\phi)$ be the covariance matrix after the rotation $\phi$ in the class $\C(S)$, obtained according to (\ref{Vrot}). Then
\begin{equation}
	S_{\hbf{rot}}(\phi)(E^\TT V(p_0)E+F)S_{\hbf{rot}}(\phi)^\TT=E^\TT V(p_\phi)E+F\;.
\end{equation}
In fact
\begin{eqnarray}
  S_{\hbf{rot}}(\phi)(E^\TT V(p_0)E+F)S_{\hbf{rot}}(\phi)^\TT &=&
  S_{\hbf{rot}}(\phi)E^\TT V(p_0)E S_{\hbf{rot}}(\phi)^\TT+ S_{\hbf{rot}}(\phi) F S_{\hbf{rot}}(\phi)^\TT\nonumber\\
  &=& E^\TT S_{\hbf{rot}}(\phi) V(p_0) S_{\hbf{rot}}(\phi)^\TT E + f S_{\hbf{rot}}(\phi) I_{2N} S_{\hbf{rot}}(\phi)^\TT
\end{eqnarray}
 where $S_{\hbf{rot}}(\phi)$ verifies the condition $S_{\hbf{rot}}(\phi)S_{\hbf{rot}}^\TT(\phi)=I_{2N}$. Hence the conclusion.
\end{proof}

This model includes the most relevant cases \cite{HolevoGiovannetti} \cite{HolevoChannels}, such as:
\begin{enumerate}
\item The {\it classical noise channel}, which merely adds classical Gaussian noise to a quantum state, i.e., $E = I_{2N}$, $f \geq 0$.

\item The {\it lossy (or attenuation) channel} in which $E=\sqrt{\eta}I_{2N}$ and $F=(1-\eta)\,I_{2N}$, with $\eta<1$ so that $V\;\to\; \eta\,V + (1-\eta) I_{2N}$. This is the model for example for the propagation along an optical fiber, where each photon is lost with probability $(1 - \eta)$. 

\item The {\it amplification channel} in which $E=\sqrt{\eta}I_{2N}$ and $F=(\eta-1)\,I_{2N}$, with the gain $\eta>1$, so that $V\;\to\; \eta\,V + (\eta-1) I_{2N}$. 

\item The {\it thermal noise channel} called sometimes also {\it attenuation channel} \cite{HolevoGiovannetti}, with $E=\sqrt{\eta}I_{2N}$ and $F=(1-\eta)\sigma^2\,I_{2N}$, with $\eta<1$ and $\sigma^2\, I_{2N}$, $\sigma^2\geq 1$, is the covariance matrix of a thermal state with average photon number $\C(N)=(\sigma^2 - 1)/2$.
\end{enumerate}
{\bf Remark} Note that in these channels the assumption is that the parameters are the same for all the $N$ modes. In particular, for the thermal noise channel, in all the modes the average number of thermal photons is considered the same. In general, denoting by $\sigma^2_{k}$ the thermal contribution in the $k$-mode, the matrix $F$ in the covariance relation should be modified as
  $$
    F=(1-\eta)\,\bigoplus_{k=1}^N\sigma^2_{k}\;I_{2}\;.
  $$
However, the assumption $\sigma^2_k=\sigma^2$ is acceptable for PPM or other modulations in quantum communications.

\section{Examples of Applications}\label{Applications}

In this section we recall the quantum detection based on the square-root measurements (SRM) in general and then in the presence of GUS. Finally we give an explicit application to PPM.

\subsection{The SRM in general}
         
In the case of pure states, the measurement vectors $| \mu_i \rangle$ are chosen with the criterion of making the differences 
between the states and the measurement vectors, $| e_i \rangle=|\gamma_i\rangle-| \mu_i \rangle$, as small as possible and we look for {\bf the measurement vectors $| \mu_i \rangle$ which minimize the quadratic error} \cite{Forney1}
\begin{equation}
\C(E) = \sum_{i=0}^{K-1} \langle e_i | e_i \rangle =
         \sum_{i=0}^{K-1} (\langle \gamma_i|-\langle\mu_i|) 
	    (|\gamma_i\rangle-\mu_i \rangle) 
\end{equation}                                                        
with the constraint of the resolution of the identity $\sum_{i=0}^{K-1}|\mu_i\rangle\langle\mu_i|=I_\C(H)$.

The evaluation of the measurement vector  is obtained computing the inner products $\langle\gamma_i|\gamma_j\rangle$ between the states of the constellation, thus obtaining the {\it Gram's matrix} 
\begin{equation}
\sottovar(2pt){G}_{K\times K}=\Gamma^*\,\Gamma=\left[\bk(\gamma_{i},\gamma_{j})  \right]\;,\qquad i,j=0, \ldots, K-1
\label{C20X}
\end{equation}
Then we evaluate, by eigendecomposition, the square root and the inverse square root $G^{\pm 1/2}$.
The measurement vectors are given explicitly by 
\begin{equation}
    |\mu_i\rangle =\sum_{j=0}^{K-1} {(G^{-1/2})}_{ij} |\gamma_j\rangle\,. \label{C27}
\end{equation}
The transition probabilities result in $p_c(j|\;i)=\left|({G}^{\frac{1}{2}})_{ij}\right|^2$, from which we obtain the correct decision probability
\begin{equation}
	P_c=\frac{1}{K}\sum_{i=0}^{K-1} \left|({G}^{\frac{1}{2}})_{ii}\right|^2\;.  \label{C29}
\end{equation}
The advantage of the SRM method is that it gives explicit results for any constellation, that is, for any modulation format.

\subsection{The SRM in the presence of GUS }

When the constellation has the GUS, the Gram's matrix becomes {\it circulant}. In fact, the inner products result in
$G_{ij}=\langle\gamma_i|\gamma_j\rangle=\langle\gamma_0|(Q^*)^iQ^j|\gamma_0\rangle = \langle\gamma_0|Q^{j-i}|\gamma_0\rangle$
and depend upon the difference $i-j\;(\hb{mod }K)$. This property provides two advantages: 1) the performance evaluation becomes
easier with the technique of the discrete Fourier transform (DFT), and 2) the corresponding quantum detection becomes optimal \cite{Forney1} \cite{Assa}. 

The eigendecomposition of the circulant Gram matrix is simply given by $G=W^*\Ld W$, where $\Ld=\diag\;[\lambda_0,\lambda_1,\ldots,\lambda_{K-1}]$ collects the eigenvalues and $W$ is the $K\times K$ DFT matrix
$W=K^{-\met}\left[W_K^{-rs}\right]$, $r, s=0,\ldots K-1$, with $W_K:=\E^{\I 2\pi/K}$. Moreover, the eigenvalues are given by the DFT of the first row of $G$
\begin{equation}
     \lambda_p=\sum_{q=0}^{K-1}G_{0q}W_K^{-pq}\;.\label{lambda_p}
\end{equation}
The square roots of $G$ are simply obtained as $G^{\pm\frac12}=W^*\Ld^{\pm\frac12}W$ and the transition probabilities are
\begin{equation}
	p_c(j|i)=\Bigl|\frac1K\sum_{p=0}^{K-1}\lambda_p^{\frac12}W_K^{-p(i-j)}\Bigr|^2\vq i,j=0,1,\ldots,K-1\;. \label{CZ24}
\end{equation}
In particular, the diagonal transition probabilities are all equal $p_c(i|i)=\left[\frac1K\sum_{p=0}^{K-1}\lambda_p^{\frac12}\right]^2$, independent of $i$, and the correct decision probability (\ref{C29}) becomes explicitly
\begin{equation}
	P_c=\left[\frac1K\sum_{p=0}^{K-1}\lambda_p^{\frac12}\right]^2\;. \label{CZ26}
\end{equation}
It is possible to obtain the explicit expression of the measurement vectors $\ket{\mu_i}$, given by $|\mu_0\rangle=\sum_{j=0}^{K-1}{(G^{-1/2})}_{ij}|\gamma_j\rangle$.

\subsection{Extension of the SRM to mixed Gaussian states}

The SRM method can be extended to  mixed states. The preliminary step is the factorization of the density operators $\rho_i=\gamma_i\gamma_i^*$ and it can be shown \cite{Eldar3} that the measurement operators can be factored as the density operators, $P_i=\mu_i \mu_i^*$, where the rank of $\mu_i$ is the same of $\gamma_i$.

In this case the error is considered between the state factors and the measurement factors, $e_i=\gamma_i-\mu_i$. 
From the factors $\gamma_i$, we first form the Gram matrix $G=[\gamma_i^*\gamma_j]_{i,j=0,1,\ldots,K-1}$,
then from the square  roots $G^{\pm 1/2}$ we can obtain both the measurement factors $\mu_i$ and the transition probabilities in a similar way as for the pure states. 

The SRM method always leads to explicit results and, in general, provides a good overestimation of the error probability, also compared with other sub-optimal methods such as the Chernoff bound \cite{Corvaja}.
 
\subsection{Application to pulse position modulation (PPM)}

Pulse position modulation (PPM) is widely adopted in free space optical transmission, and is a candidate for deep-space
 transmission, also in quantum form. Here we evaluate the error probability in $K$--ary quantum optical PPM systems, considering the most general Gaussian states.
 
In the quantum PPM the modulation format and the states belong to a composite Hilbert space, given by the tensor product 
 $\C(H)=\C(H)_0\otimes\C(H)_0\otimes\cdots\otimes\C(H)_0$ of $K$ equal Hilbert spaces $\C(H)_0$ \cite{Yuen}\cite{CarPierPPM},
where $ \C(H)_0$ has dimension $n$ and  $\C(H)$ has dimension $N=n^K$
\begin{equation}
	|\gamma_i\rangle=|\gamma_{i,0}\rangle\otimes|\gamma_{i,1}\rangle \otimes \cdots \otimes|\gamma_{i,K-1}\rangle \vq i=0,1,\ldots K-1\;. \label{W4}
\end{equation}
Considering Gaussian states, with the most general squeezed-displaced states, symbolized by $\ket{z,\alpha}$, the natural choice for PPM is to associate the symbol 0 to the ground state $\ket{\gamma_{ik}}=\ket{0,0}$ and to the symbol 1 the generic state $\ket{\gamma_{ik}}=\ket{z,\alpha}$.
With this choice (\ref{W4}) represents a constellation of $K$-mode Gaussian states.
For instance for $K=3$ we have explicitly
\begin{eqnarray}
|\gamma_0\rangle&=&\ket{z,\alpha}\otimes\ket{0,0}\otimes\ket{0,0}\nonumber \\
|\gamma_1\rangle&=&\ket{0,0}\otimes\ket{z,\alpha}\otimes\ket{0,0}\\ \label{W5}
|\gamma_2\rangle&=&\ket{0,0}\otimes\ket{0,0}\otimes\ket{z,\alpha}\nonumber \;.
\end{eqnarray}
Note that, without loss of generality we can choose a real displaced parameter $\alpha$, while we let a generic complex squeezing factor $z=r\,e^{i \theta}$.

The application of GUS to PPM is not trivial because the states are multimode. The symmetry operator $Q$ is given by \cite{CarPierPPM}--\cite{Henderson}
\begin{equation}
  Q=\sum_{k=0}^{n-1}w_n(k)\otimes I_{N'}\otimes w^\TT_n(k)\;,\qquad N'= n^{K-1}\;,\label{P2A}
\end{equation}
where $\otimes$ is the Kronecker's product, $w_n(k)$ is a column vector of length $n$, with null elements except for one unitary element at position $k$ and $I_{N'}$ is the $N'\times N'$ identity matrix. Then $Q$ has dimension $N=n^K$ and the property $Q^K=I_N$.
  
Now, it is not immediate to see that $Q$ is a rotation operator, that is, of the form $R(\phi)=\E^{\I\,\phi}$, with $\phi$ an $N\times N$ Hermitian matrix. To find the ``phase'' $\phi$ we use the EID of $S$ written in the form
$ Q=\sum_{m=0}^{K-1} \lambda_m P_m$   
where $\lambda_m=\E^{\I\,2\pi m/K}:=W_K^m$ are the $K$ distinct eigenvalues and $P_m$ are $K$ orthogonal projectors.
In this EID the eigenvalues are known, while the projectors should be evaluated from the the expression (\ref{P2A}), which defines a complicate permutation matrix. The alternative is the evaluation through the powers of $Q$, $Q^k=\sum_{m=0}^{K-1} W_K^{mk} P_m$.
According to this relation, $[Q^0,Q^1,\ldots,Q^{K-1}]$ turns out to be the DFT of  $[P_0,P_1,\ldots,P_{K-1}]$. Thus, taking the inverse DFT one gets $P_m=\frac1K\sum_{k=0}^{K-1} W_K^{-mk} Q^k$, which is easy to evaluate. Next we recall that $Q$ is unitary and therefore it can be written in the form $Q=\E^{\I\,\phi}$, where $\phi$ is a Hermitian matrix. Then, by comparison with the EID of $Q$ we find that the EID of $\phi$ is given by
\begin{equation}
  \phi=\sum_{m=0}^{K-1} \frac{2\pi m}{K} P_m\;, \label{HK20}
\end{equation}
where the eigenvalues become $2\pi\, m/K$ and the projectors are the same as in the EID of $Q$.
 
\noindent {\bf Example} We give an example with $K=3$ and $n=2$, where the matrices are $8\times 8$. The symmetry
operator is
\begin{equation}
	Q=\left[\begin{array}{cccccccc}
		 1 & 0 & 0 & 0 & 0 & 0 & 0 & 0 \\
		 0 & 0 & 0 & 0 & 1 & 0 & 0 & 0 \\
		 0 & 1 & 0 & 0 & 0 & 0 & 0 & 0 \\
		 0 & 0 & 0 & 0 & 0 & 1 & 0 & 0 \\
		 0 & 0 & 1 & 0 & 0 & 0 & 0 & 0 \\
		 0 & 0 & 0 & 0 & 0 & 0 & 1 & 0 \\
		 0 & 0 & 0 & 1 & 0 & 0 & 0 & 0 \\
		 0 & 0 & 0 & 0 & 0 & 0 & 0 & 1 
              \end{array}\right]
\label{QPPM}
\end{equation}
The projectors $[P_0, P_1, P_2]$ are obtained by the DFT of $[Q^0,Q^1,Q^2]$ and finally we have the phase matrix $\phi$ from (\ref{HK20}). It reads
\begin{small}
\begin{equation}
\phi=\frac{3}{\pi}\left[
\begin{array}{cccccccc}
0 & 0 & 0 & 0 & 0 & 0 & 0 & 0 \\
0 & 2 & -1-\frac{i}{\sqrt{3}} & 0 & -1+\frac{i}{\sqrt{3}} & 0 & 0 & 0 \\
0 & -1+\frac{i}{\sqrt{3}} & 2 & 0 & -1-\frac{i}{\sqrt{3}} & 0 & 0 & 0 \\
0 & 0 & 0 & 2 & 0 & -1+\frac{i}{\sqrt{3}} & -1-\frac{i}{\sqrt{3}} & 0 \\
0 & -1-\frac{i}{\sqrt{3}} & -1+\frac{i}{\sqrt{3}} & 0 & 2 & 0 & 0 & 0 \\
0 & 0 & 0 & -1-\frac{i}{\sqrt{3}} & 0 & 2 & -1+\frac{i}{\sqrt{3}} & 0 \\
0 & 0 & 0 & -1+\frac{i}{\sqrt{3}} & 0 & -1-\frac{i}{\sqrt{3}} & 2 & 0 \\
0 & 0 & 0 & 0 & 0 & 0 & 0 & 0
\end{array}
\right]\label{PhiPPM}
\end{equation}
\end{small}
We can verify (e.g. with {\tt Mathematica}) that $\E^{\I\,\phi}=Q$ and that $\E^{\I\,K\,\phi}=I_N$.

\subsection{Statistics of the quantum states in PPM}

To evaluate the error probability of QC systems with PPM we need the following statistics of the single--mode state $\ket{z,\alpha}$:
\begin{itemize}
\item The mean photon number, which is given by \cite{Yuen2}
\begin{equation}
  \bar N_{|z,\alpha\rangle}=|\alpha|^2+\sinh^2(r) \label{U22A}
\end{equation}
Then all the $K$ PPM symbols have the same {\it mean number of photons per symbol}, given by
\begin{equation}
	N_s=|\alpha|^2+\sinh^2(r)\;. \label{R2}
\end{equation}
\item The inner product between two states was evaluated by Yuen \cite{Yuen2} and reads
\begin{equation}
 \bk({z_1,\alpha_1},{z_0,\alpha_0})= A^{-\met}\;\exp\left[-\frac{A\left(|\beta_1|^2+|\beta_0|^2\right)-
              2\beta_1\beta_0^*+B\,\beta_1^{*2}-B^*\beta_0^2}{2A}\right]
	    \label{YY2}  
\end{equation}
where  $\mu_i=\cosh(r_i)$, $\nu_i=\sinh(r_i)\E^{\I\theta_i}$, $\beta_i=\mu_i\alpha_i-\nu_i\alpha_i^*$,
$A=\mu_0\mu_1^*-\nu_0\nu_1^*$, $B=\nu_0\mu_1-\mu_0\nu_1$.
\end{itemize}

\subsection{Error probability in the quantum PPM}

The analysis of a quantum PPM  system (limited to coherent states) has been done  in a famous article by Yuen, Kennedy and Lax \cite{Yuen} who found the optimal elementary projectors using an algebraic method developed ``ad hoc'' for this kind of modulation. In \cite{CarPierPPM} we proposed an original method based on the SRM which given the minimum error probability for the GUS of the quantum PPM.
Also in \cite{CarPierPPM} the analysis was limited to coherent states. Here we extend the evaluation to general Gaussian states.

In the SRM  the error probability depends only on the inner product between the single--mode states $\ket{z,\alpha}$ and $\ket{0,0}$. In fact the Gram matrix is given by
\begin{equation}
	G=\qmatrix{1 & \Gamma & \ldots & \Gamma\cr
                    \Gamma & 1 &  \ldots & \Gamma\cr
                     \vdots &  & \ddots & \cr
                     \Gamma & \Gamma & \ldots & 1\cr}\label{Gram1}
\end{equation}
where $\Gamma:=|\bk({z,\alpha},{0,0})|^2$.  

\subsubsection{About the inner product}

The squared inner product can be written in the form
\begin{equation}
   \Gamma=\frac{1}{\cosh r}\exp[-\alpha^2\;f(r,\theta)]\label{GG1}
\end{equation}
where
\begin{equation}
  f(r,\theta)= \cosh(2r)+ \tanh(r)\sinh(2r) - \left[\sinh(2r)+\tanh(r)\cosh(2r)\right]\,\cos\theta\ \;.\label{FR2}
\end{equation}
Clearly, for $r$ and $\alpha$ given, $\Gamma$ has a minimum for $\theta=\pi$, as shown in Fig.~\ref{GQ17}.
\begin{figure}[h]
\includegraphics[width=\textwidth]{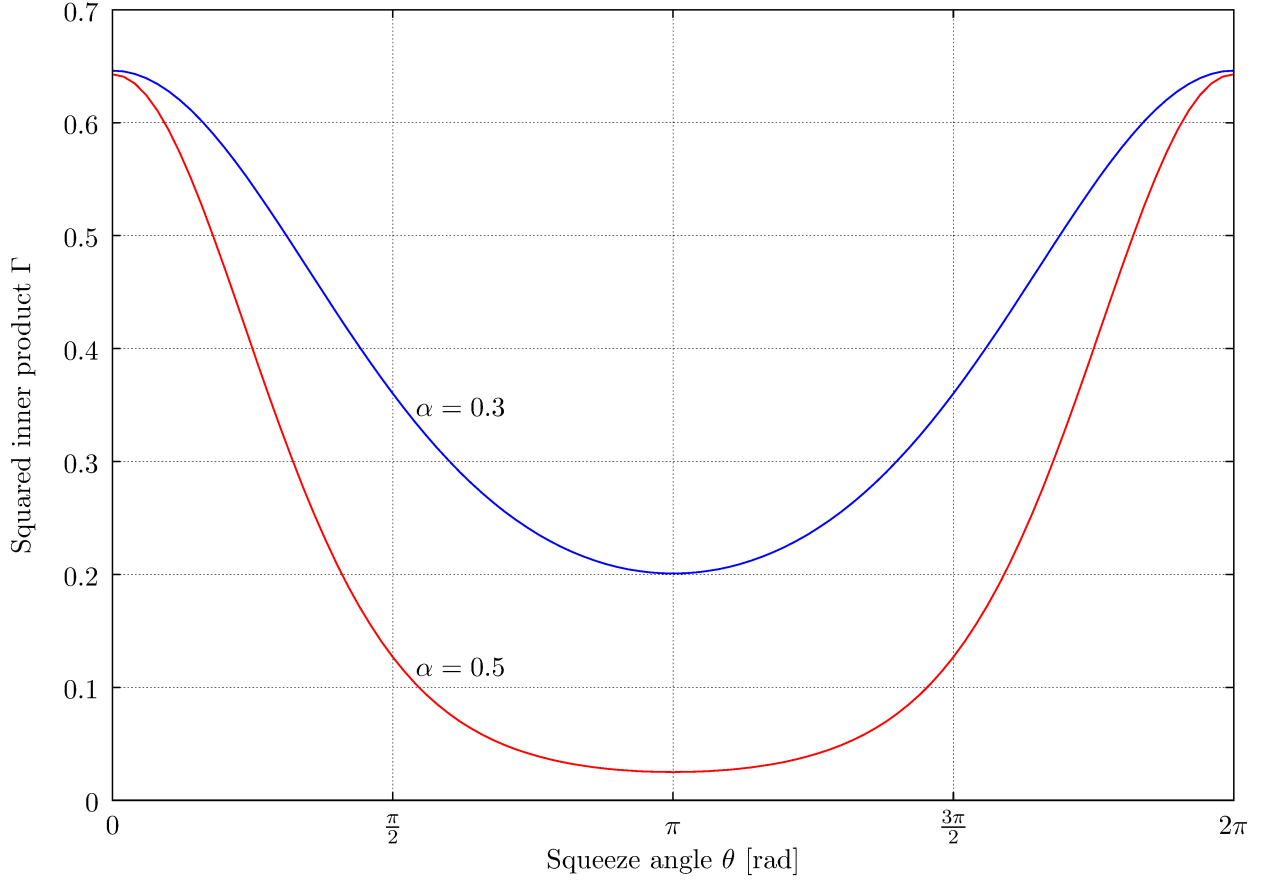}
\caption{Squared inner product $\Gamma=|\langle r\;e^{i\theta},\alpha|0,0\rangle|^2$ versus the phase $\theta$ for $r=1.0$ and two values of the displacement.}\label{GQ17}
\end{figure}
The interpretation in the phase space may be the following. In the phase space, thinking to the Wigner function, an inner product as
$\bk({r\,e^{i\theta},\alpha},{0,0})$ depends on the distance between the two states, which is provided by the displacement $\alpha$, but also on the ``orientation'' of the squeezing, which is determined by the phase $\theta$. We know that in a squeezed displaced state $\ket{r\,e^{i\theta},\alpha}$ the noise variances are different and pictorially this difference can be represented by a tilted ellipse, as shown in Fig.~\ref{GQ14}.
\begin{figure}[h]
\includegraphics[width=\textwidth]{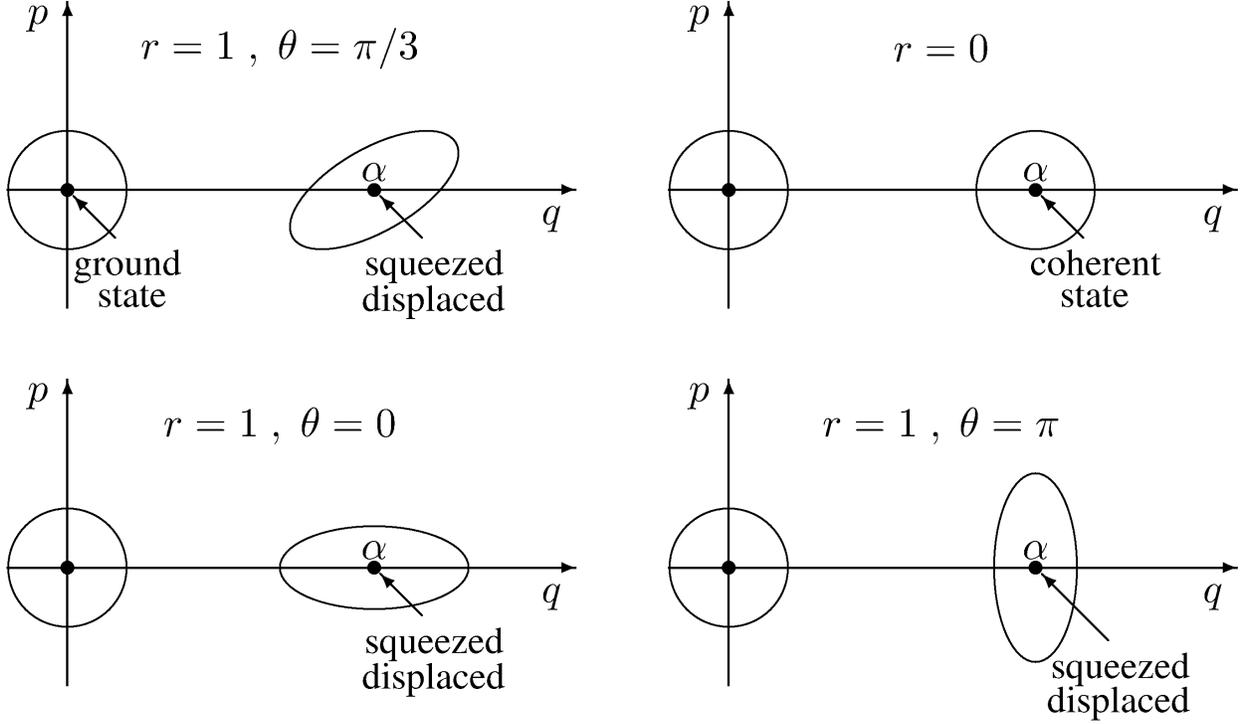}
\caption{Representation of the states $\ket{z,\alpha}$ and $\ket{0,0}$ in the phase space for different
values of $z=r\,e^{i\theta}$. Pictorially the noise variances 
are represented by a tilted ``error ellipse''.}\label{GQ14}
\end{figure}
The ellipse degenerates into a circle in the case of a coherent state (and for the ground state). 
With the objective to minimize the error probability, the minimum value of $\Gamma$ is sought. Form the figure we can easily realize that a squeezing factor with $\theta = 0$ gives a worst error probability than the use of a coherent state, while the best performance is achieved with $\theta = \pi$.

The error probability computed from (\ref{C29}) becomes \cite{CarPierPPM}
\begin{equation}
  P_{e} = 1-\frac{1}{K^2}\left(\sqrt{1+(K-1)\Gamma}+(K-1)\sqrt{1-\Gamma}\right)^2 \label{W30}
\end{equation}
in perfect agreement with the results of \cite{Yuen}. Also in this case $P_e$ can be expressed in terms of the mean photon number per symbol $N_s$, by writing $\Gamma$ as
\begin{equation}
  \Gamma=\frac{1}{\cosh r}\exp\left[-\left(N_s-\sinh^2 r\right)\;f(r,\theta)\right]\vq N_s\ge \sinh^2 r\;.\label{Gamma1}
\end{equation}
In the representation of the error probability it is convenient to consider as a variable the {\it average number of photons per bit} $N_R=N_s/\log_2 K$. We see in the expression of $N_s$ (\ref{R2}) that a contribution comes from the displacement and one form the squeezing. If we fix a value of $r$, the minimum of $N_s$ becomes $\sinh^2(r)$ and for $N_s<\sinh^2(r)$ there is no room for the displacement.

In Fig.~\ref{8PPMquantum} we present the error probability as a function of $N_R$ for coherent states and other three values of the 
squeezing factor.  
\begin{figure}
\includegraphics[width=\textwidth]{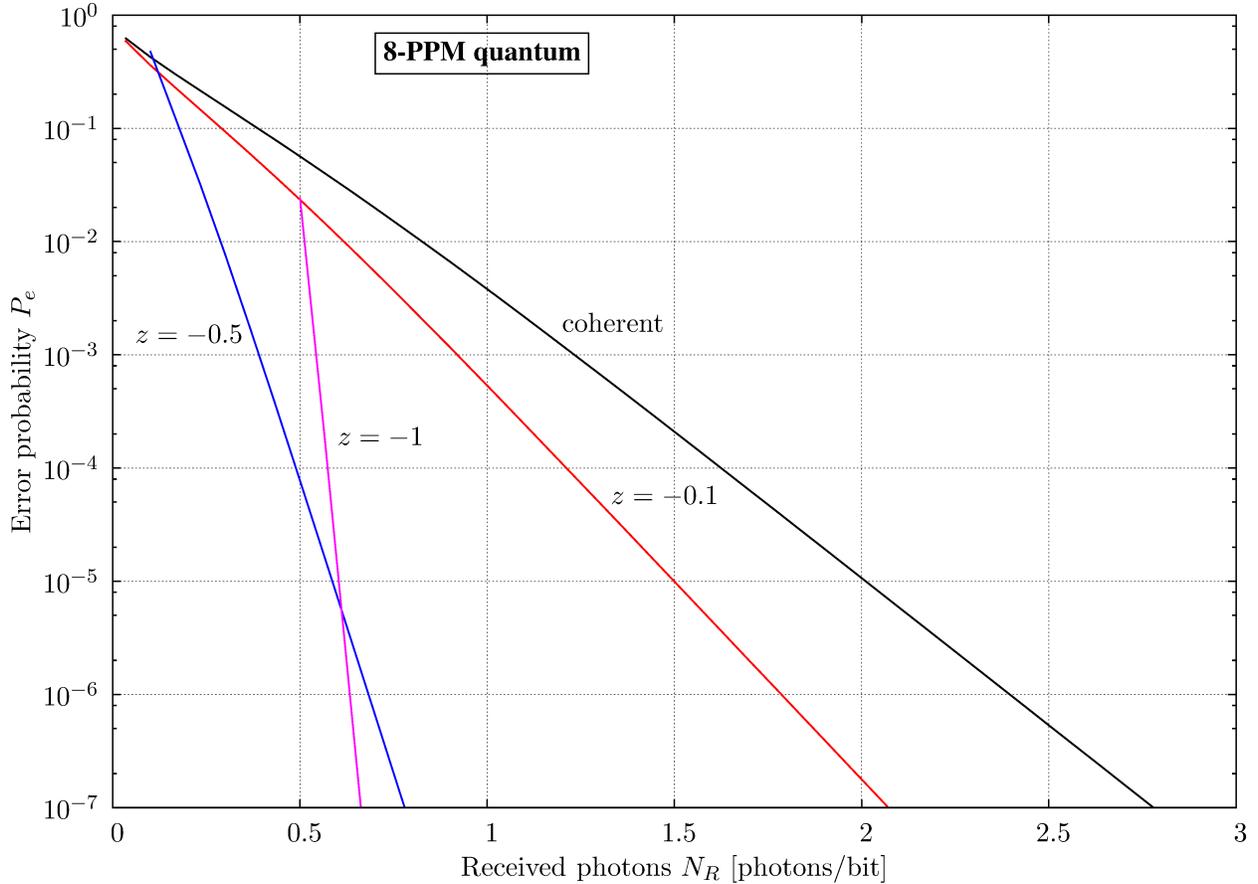}
\caption{Error probability of quantum 8--PPM for coherent states and three values of squeezing. For $z=-0.1$ we must have $N_R\ge 3.3 \cdot 10^{-3}$, for $z=-0.5$, $N_R\ge0.09$ and for $z=-1$, $N_R\ge0.46$. }\label{8PPMquantum}
\end{figure}
Note that, by properly choosing the value of $z$, one can reduce dramatically the error probability with respect to the use of coherent states. 

\section{Conclusions}\label{Conclusions}

We have seen that GUS plays a fundamental role in quantum communications, due to the optimality of measurement operators obtained by the SRM in the quantum discrimination. 
The considerations on Gaussian states and their invariance properties with respect to unitary transformations and in particular rotations allow one to construct constellations of Gaussian states having the GUS, for example coherent states for their use in quantum optical communications. Moreover, the transmission of such states through and additive-noise channel preserves the GUS.
The theory of the GUS applied to the most general Gaussian states extends the analysis of the performance of a QC system employing PPM (or other modulations with GUS) to the most general case, not limiting the evaluation to the case of coherent states.

\begin{acknowledgments}
This work has been supported in part by the Project ``Q-FUTURE'' (prot.~STPD08ZXSJ) of the University of Padova.
\end{acknowledgments}


\begin{thebibliography}{99}
%
\bibitem{Weedbrook} C.~Weedbrook, S.~Pirandola, R.~Garc\'{\i}a-Patr\'on, N.J.~Cerf, T.C.~Ralph, C.J.H.~Shapiro, and S.~Lloyd,
   ``Gaussian quantum information,''
   {\it Rev. Mod. Phys.}, Vol.~84, No.~2, pp.~621--669, May~2012.
%
\bibitem{Paris} M.G.A. Paris, 
  ``Property of squeezed number states and squeezed thermal states,''
   {\it Eur. Phys. J.}, {\bf 203}, 61,  2012.
%
\bibitem{Adesso} G. Adesso, S. Ragy, and A.R. Lee, 
   ``Continuous variable quantum information: Gaussian states and beyond,''
   {\it Open Syst. \& Inform. Dynamics}, Vol.~21, 1440001, 2014.
%
\bibitem{milanesi} A. Ferraro, S. Olivares, and M.G.A. Paris,
   ``Gaussian states in continuous variable quantum information,''
   Bibliopolis, Napoli, 2005.
%
\bibitem{Kato1999} K. Kato, M. Osaki, M. Sasaki, and O. Hirota, 
  ``Quantum detection and mutual information for QAM and PSK signals,'' 
   {\it IEEE Trans. on Comm.}, vol. 47, pp. 248--254, Feb.~1999.
%
\bibitem{CarPier} G. Cariolaro and G. Pierobon, 
  ``Performance of quantum data transmission systems in the presence of thermal noise,'' 
  {\it IEEE Trans. on Comm.}, vol.58, pp 623--630, 2010.
%
\bibitem{Ma} X. Ma and W. Rhodes, 
  ``Multimode squeeze operators and squeezed states'',
  {\it Phys. Rev. A}, Vol.~41, No.~9, pp.~4625--4631, May~1990.
%
\bibitem{Slusher} 
  E.S. Slusher and B. Yurke, 
  ``Squeezed light for coherent communications,'' 
  {\it Journal of Lightwave Techn.}, Vol. 4, no.~3, pp.~466–477, March~1990.
%
\bibitem{Hausladen} P. Hausladen and W.K. Wooters, 
   ``A `Pretty Good' Measurement for Distinguishing Quantum States,'' 
   {\it Journal of Modern Optics}, Vol.~41, No.~12, pp.~2385--2390, 1994.
%
%
\bibitem{Forney1} Y.C.~Eldar and G.D. Forney, Jr., 
   ``On quantum detection and the square-root measurement,''
   {\it IEEE Trans. on Inform. Theory}, Vol.~47, No.~3, pp.~858--872, Mar.~2001.
%
\bibitem{Eldar3}  Y.C. Eldar, A. Megretski, and G.C. Verghese, 
  ``Optimal detection of symmetric mixed quantum states,''
  {\it IEEE Trans. on Inform. Theory}, vol IT--50, pp. 1198--1207, June 2004.
%
\bibitem{Corvaja} R. Corvaja,
   ``Comparison of error probability bounds in quantum state discrimination,'' 
   {\it Phys. Rev. A}, Vol.~87, No.~4, Apr.~2013.
%
\bibitem{HolevoGiovannetti} A.S. Holevo and V. Giovannetti
   ``Quantum channels and their entropic characteristics,''
   {\it Rep. on Progress in Physics}, Vol.~75, No.~4, Apr.~2012. 
%
\bibitem{Shumaker86} B.L. Schumaker, 
   ``Quantum mechanical pure states with Gaussian wave functions,'' 
   {\it Phys.  Rep.}, vol.~135, No.~6, p.~317--408, Apr.~1986.
%
\bibitem{Ma89} X. Ma,
  ``Time evolution of stable squeezed states,''
  {\it J. Mod. Optics}, Vol.~36, No.~8, pp.~1059--1064, 1989.
%
\bibitem{Helstrom2} C.W. Helstrom, 
   {\it Quantum detection and estimation theory}.
   New York: Academic Press, 1976.
%
\bibitem{Kennedy} R.S. Kennedy, 
   ``A near--optimum receiver for the binary coherent state quantum channel,'' 
   {\it MIT Research Laboratory of Electronics Quartely Progress Report 108}, Cambridge, pp. 219--225, Jan.~1973.
%
\bibitem{Yuen} H.P. Yuen, R.S. Kennedy, and M. Lax, 
  ``Optimum testing of multiple hypotheses in quantum detection theory,'' 
   \emph{IEEE Trans. on Inform. Theory}, vol. IT-21, no.~2, pp.~125--134, March~1975.
%
\bibitem{Ban1997} M. Ban, K. Kurokawa, R. Momose, and O. Hirota,
 ``Optimum measurement for discrimination among symmetric quantum states and parameter estimation,'' 
  {\it Int. J. of Theor. Phys.}, vol.~36, pp.~1269--1288, 1997.
%
\bibitem{HolevoChannels} A.S. Holevo and R.F. Werner,
  ``Evaluating capacities of bosonic Gaussian channels,'' 
  {\it Phys. Rev. A}, Vol.~63, 032312, 2001.
%
\bibitem{Assa} A. Assalini, G. Cariolaro, and G. Pierobon, 
  ``Efficient optimal minimum error discrimination of symmetric quantum states,''
  {\it Phys. Rev. A}, Vol.~81, 012315, 2010.
%
\bibitem{CarPierPPM} G. Cariolaro and G. Pierobon, 
  ``Theory of quantum pulse position modulation and related numerical problems,'' 
  {\it IEEE Trans. on  Comm.}, Vol.~58, No.~4, pp.~1213--1222, April~2010.
%
\bibitem{Henderson} H.V. Henderson and S.R. Searle, 
   ``The vec permutation matrix, the vec operator and Kronecker products: a review,''
   {\it Linear and Multivariate Algebra}, vol.~9, pp.~271--288, Jan.~1981. 
%
%
\bibitem{Yuen2} 
  H.P.~Yuen, 
  ``Two-photon coherent states of the radiation field,''
  \emph{Phys. Rev. A}, vol.~13, No.~6, p.~2226, June~1976.
\end{thebibliography}
\end{document}